\documentclass[11pt]{amsart}

\usepackage{xcolor}
\usepackage{amsmath, amssymb, graphicx}
\usepackage{amsthm}
\usepackage[capitalise,noabbrev]{cleveref}

\newcommand{\ruleset}[1]{\textsc{#1}}
\newcommand{\cclass}[1]{\ensuremath{\mathord{\rm #1}}}

\theoremstyle{plain}
\newtheorem{openProblem}{Open Problem}

\newtheorem{corollary}{Corollary}
\newtheorem{proposition}{Proposition}
\newtheorem{lemma}{Lemma}

\theoremstyle{definition} \newtheorem{definition}{Definition}

\makeatletter \let\cl@chapter\relax \makeatother

\newcommand{\gDistance}[2]{\ruleset{Graph\-Dis\-tance\-(\ensuremath{#1,#2})}}

\newcommand{\ensnort}[1]{\ruleset{EnSnort\ensuremath{\left(#1\right)}}}

\begin{document}

\title{Keeping Your Distance is Hard} 

\author{Kyle Burke}
\address{Plymouth State University, USA}
\email{kwburke@plymouth.edu}

\author{Silvia Heubach}
\address{California State University Los Angeles, USA}
\email{sheubac@calstatela.edu}

\author{Melissa Huggan}
\address{Dalhousie University, Canada}
\email{melissa.huggan@dal.ca}

\author{Svenja Huntemann}
\address{Dalhousie University, Canada}
\email{svenja.huntemann@dal.ca}


\thanks{Thanks to Richard Nowakowski and AARMS for making Games at Dal 2015 at Dalhousie University possible, where most of this work came from. The first author's visit at this conference was supported by the Plymouth State University Offices of the Provost and the Dean of Arts and Sciences. The third and fourth authors' research was supported in part by the Natural Sciences and Engineering Research Council of Canada and the Killam Trust.}

\begin{abstract}
We study the computational complexity of distance games, a class of combinatorial games played on graphs. A move consists of placing a coloured token on an unoccupied vertex   subject to it  not being at certain distances to already occupied vertices. The last player to move wins.  Well-known examples of distance games are \ruleset{Node-Kayles}, \ruleset{Snort}, and \ruleset{Col}, whose complexities were shown to be \cclass{PSPACE}-hard. We show that many more distance games are also \cclass{PSPACE}-hard.
\end{abstract}
\maketitle

\section{Introduction}

We begin by introducing distance games and defining specific combinatorial games needed in the remainder of the paper, then give an introduction to computational complexity and explain our proof strategy. At the end of the section we will give an overview of the organization of the paper.

\subsection{Distance Games}

Distance games were introduced by Huntemann and Nowakowski \cite{HuntemannN2014}. They are part of a larger class of combinatorial games called \textit{placement games} studied in \cite{BrownCHMMNS} and  \cite{FaridiHN}. 
Distance games are played by placing game pieces (tokens) on empty vertices of a graph (= game board) according to rules that forbid game pieces to be placed at certain distances to already occupied vertices. 
The \textit{distance} between two game pieces
is defined as the graph distance between the respective vertices they occupy. 
We make the rules of the game precise in the following definition.

\begin{definition}
  \gDistance{D}{S} is the combinatorial game played on an arbitrary graph $G = (V,E)$, where $S$ (stands for ``same") and $D$ (stands for ``different") are subsets of $\{1, 2, \ldots, n\}$, where $n=|V|$.  Each graph vertex may be empty or contain a blue or red piece.  Two players, Blue and Red, take their turn by placing a single piece of their colour on an empty vertex according to these rules:
  \begin{enumerate}
     
          \item A blue piece and a red piece are not allowed to have distance $d\in D$.
          \item Two  pieces of the same colour are not allowed to have distance  $s\in S$.

  \end{enumerate}
  Pieces may neither be moved to another vertex nor be removed once placed. The game ends when one player can no longer place a piece of their colour. 
\end{definition}

For ease of readability, we will sometimes refer to the action of placing a coloured piece on an unoccupied vertex as ``colouring the vertex" and we will call an unoccupied vertex ``uncoloured". Also note that, since we play on a finite graph of $n$ vertices,  the diameter of the graph is at most $n-1$, so $S$ and $D$ can be restricted to be finite without any loss of generality.\\ 


By definition, \gDistance{D}{S} is a \textit{partizan} combinatorial game, that is, the two players have different moves available. This results from the fact that their game pieces have different colours. For example, in the game \gDistance{\{1\}}{\varnothing} played from the position shown in Figure~\ref{fig:impartial}, Blue cannot play on vertices $2$ or $6$, while Red is allowed to play on all vertices.

\begin{figure}[h!]
  \begin{center}\includegraphics[scale=.40]{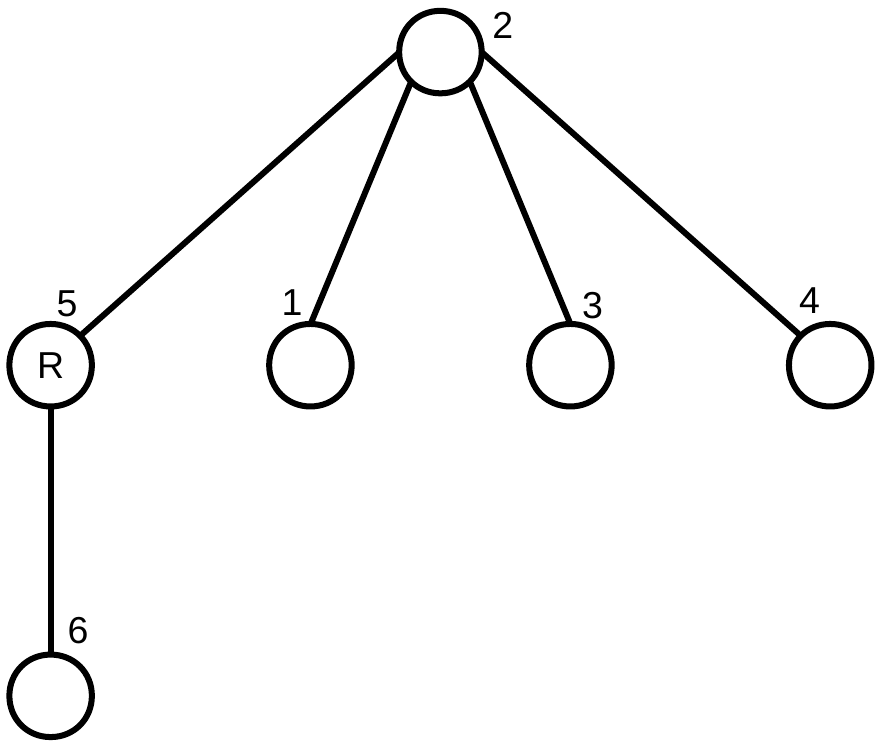}\end{center}
  \caption{A position in a graph distance game. A red piece has been placed on vertex five, while all other vertices do not yet have a piece placed on them. }
  \label{fig:impartial}
\end{figure}

On the other hand, if $D = S$, then the colour of the game pieces becomes irrelevant, allowing both players to play on the same set of vertices (though with different coloured pieces). 
Consider again the position shown in Figure~\ref{fig:impartial}. If $D=S=\{1\}$, then Blue cannot play on vertices $2$ or $6$ because $D=\{1\}$, while Red cannot play on these vertices because $S=\{1\}$. Thus, when $D = S$,  \gDistance{D}{S} is equivalent to an \textit{impartial} combinatorial game, one where both players can make the same moves and the pieces have only one colour.

The games listed below are well-known combinatorial games that will play a key role in the derivation of the complexity of distance games. We define them here for the reader unfamiliar with these games and express them as graph distance games.


\begin{definition}
In the game of \ruleset{Node-Kayles}, players alternate placing tokens of a common colour on empty vertices, all of whose neighbours are also empty. The game ends when no tokens can be placed.  
\end{definition}

Thus the impartial game \ruleset{Node-Kayles} corresponds to  \gDistance{\{1\}}{\{1\}}. 


\begin{definition}
The game \ruleset{BiGraph-Node-Kayles} is a partizan version of \ruleset{Node-Kayles} played on a bipartite graph with vertex partition $V=V_B \cup V_R$, with the additional restriction that Blue can only play on $V_B$ and Red can only play on $V_R$. 
\end{definition}


Even though the game pieces in \ruleset{BiGraph-Node-Kayles} have the same colour, the fact that all vertices within $V_B$ and $V_R$ are at even distances, and those from different vertex sets are at odd distances,  \ruleset{BiGraph-Node-Kayles} corresponds to \gDistance{D}{S}, where $S$ is the set of odd integers and $D=\{1\}\cup \{ 2, 4, 6, \ldots\}$, if at least one vertex is coloured. (If no vertices are coloured initially, then there is no way to enforce that the first player plays on the appropriate side of the bipartite graph.)


\begin{definition}\label{def:SnortCol}
In the partizan games of \ruleset{Snort} and \ruleset{Col}, players place red and blue tokens on vertices of a graph on any empty vertex with the restriction that adjacent vertices cannot have tokens of different or the same colour, respectively. The game ends when no tokens can be placed.
\end{definition}

Thus \ruleset{Snort} is the game \gDistance{\{1\}}{\varnothing} and    \ruleset{Col} is the game \gDistance{\varnothing}{\{1\}}. More information, such as simple winning strategies and other properties of \ruleset{Node-Kayles}, \ruleset{Snort}, and \ruleset{Col} can be found in \cite{WinningWays} and \cite{Schaefer1978}.






In our discussion about the complexity of determining winnability, we need to take into account all the possible positions of the game, together with the rules for each player. We will refer to this larger structure as a \emph{ruleset}. Since $S$ and $D$ completely determine the available moves, the triple $(Q, S, D)$, where $Q$ is the set  of positions of the game, specifies the ruleset. 




\subsection{Computational Complexity of Games}

Computational complexity can be applied to combinatorial games to measure how hard it is to determine whether the next player has a winning strategy.  An algorithm is considered \emph{efficient} if the running time can be expressed as a polynomial of the size of the input (i.e.~the game to be analyzed).  The computational class \cclass{P} is the set of all computational problems with these polynomial-time solutions.  \cclass{EXPTIME} consists of all problems that can be solved in exponential time or less.  Problems that \emph{require} an exponential amount of time make up the class \cclass{EXPTIME}-hard.

For the ruleset of any combinatorial game, the problem to be answered is always: ``For a given position, $G$, that has some or no coloured vertices, does Blue, moving next, have a winning strategy?"  Note that this does not mean we need to find the winning strategy, just determine whether one exists.  
For some rulesets, such as \ruleset{Nim} \cite{Bouton:1901}, this problem is in \cclass{P}.  Some other games, such as a version of \ruleset{Chess} generalized to an $n \times n$ board, are instead \cclass{EXPTIME}-hard~\cite{DBLP:journals/jct/FraenkelL81}, meaning the problems are \emph{intractable}.

\cclass{PSPACE} is the set of decision problems that can be solved using a polynomial amount of storage with no restrictions on time.  It is known that $\cclass{P} \subseteq \cclass{PSPACE} \subseteq \cclass{EXPTIME}$, and $\cclass{P} \subsetneq \cclass{EXPTIME}$, but it is not known whether either of the first two inclusions are strict or not 
\cite{PapadimitriouBook:1994}. Many decision problems are \cclass{PSPACE}-\emph{hard}, meaning they are at least as hard as the most difficult problems in \cclass{PSPACE}. 
A decision problem is \cclass{PSPACE}-\emph{complete} if it is both \cclass{PSPACE}-hard and in \cclass{PSPACE}.

The input for the decision problem is the triple $(Q, S, D)$ which specifies the ruleset. Since we play on a graph with $n$ vertices,  each position has at most $n$ move options for the current player.  Furthermore, since each move places a token on an unoccupied vertex, the game ends in at most $n$ moves, so the game tree for a given position has height at most $n$. 
To analyze the decision tree, we only need to keep track of the options of any given position that is being analyzed. So at worst we need to store $n$ positions (maximal length of the game) with at most $n$ options each, so the storage space is of order $n^2$. Therefore, distance games on graphs are at worst in PSPACE. 
As a result, any \cclass{PSPACE}-hard distance game is also \cclass{PSPACE}-complete: it is among the hardest problems included in \cclass{PSPACE}. 

A ruleset $T'$ can be shown to be \cclass{PSPACE}-hard with the help of another ruleset, say $T$, already known to be \cclass{PSPACE}-hard.  $T'$ is \cclass{PSPACE}-hard if a function $f$ exists where
\begin{itemize}
  \item $f:$ positions$(T) \rightarrow $ positions$(T')$,
  \item $f$ can be computed in polynomial time, and
  \item $f$ preserves winnability (e.g. for $t \in \text{positions}(T)$, Blue has a winning move going next on $f(t)$ exactly when Blue has a winning move going next on $t$) \cite{AlgGameTheory_GONC3}.
\end{itemize}

Such a function $f$ is called a \emph{reduction} (from $T$ to $T'$).  Finding reductions from \cclass{PSPACE}-hard games to new games is common practice for showing the \cclass{PSPACE}-hardness of these new games.  Sometimes these reductions have a stronger property: each move from any position $t \in T$ corresponds to exactly one move in $f(t)$, that is, the game trees have exactly the same shape.  Due to this injective homomorphism, we can refer to these reductions as \emph{play-for-play} reductions.  Note that the complexity results for games consider the complete set of positions, not just starting positions, because all these positions contribute to the analysis of winnability of a game.  

There are many theorems and conjectures about the winnability of games from their starting positions, though these are usually non-constructive.  Just like reductions for the hardness of \ruleset{Hex}~\cite{Reisch:1981}, \ruleset{Snort}~\cite{Schaefer1978}, and others, our reductions use partially-coloured graphs. Readers interested further in the application of computational complexity to combinatorial games should reference \cite{AlgGameTheory_GONC3}.

\subsection{Reduction Strategy}
In what follows, we will describe each reduction as a transformation of the graph $G$, on which a game $T$ is played, to  a graph $G'$, on which game $T'$ is played, via the insertion of subgraphs called {\em gadgets}. All reductions to be used will be play-for-play, as we will enforce the following two properties in all of our constructions:

\begin{itemize}
  \item \textbf{Vertex condition}: No vertex added to $G$ to form $G'$ is playable. No vertex of the original graph $G$ is deleted.
  \item \textbf{Edge condition}: None of the additional edges will result in any restrictions on the play on any of the  vertices $v \in V$ from the original graph $G$. That is, for any $v \in V$,  a Blue/Red piece can be played at $v$ under ruleset $T$ on $G$ exactly when it can be played on $v$ using ruleset $T'$ on $G'$.
\end{itemize}

We will use the fact that \ruleset{Node-Kayles}, \ruleset{BiGraph-Node-Kayles}, \ruleset{Snort}, and \ruleset{Col} (all placement games played on graph) 
are \cclass{PSPACE}-hard \cite{Burke2018, Schaefer1978} to create reductions showing that the games \gDistance{D}{S} are \cclass{PSPACE}-hard for many pairs $D$ and $S$. In all the reduction proofs, the goal is to create a graph $G'$ via a play-for-play reduction so that playing \gDistance{D}{S} on $G'$ is equivalent to playing the game from which we reduce on $G$. We let $V'$ be the union of all inserted vertices, $E'$ be the union of all inserted edges, $E_1$ the set of edges of $G$ not replaced, and $G'=(V\cup V', E_1\cup E')$ be the desired graph on which to play \gDistance{D}{S}. 
In our diagrams, we will use B (blue) and R (red) for the pieces of the respective players.

\subsection{Outline}


Even though the simplest case, namely $D = \{1, 2\}$, is covered by the more general cases to follow, we will use it  to give an introduction to the construction of gadgets by reducing from \ruleset{Node-Kayles} using very simple gadgets in Section~\ref{sec:Dsmall}. We develop a more complex gadget in Section \ref{sec:gadgets} that will be used for more general sets $S$ and $D$ in Sections~\ref{sec:max not same} ($\max(S) \neq \max(D)$) and  ~\ref{sec:same max} ($\max(S) = \max(D)$). We apply results on PSPACE-hardness of planar \ruleset{Col} and planar \ruleset{Snort} to our results 
in Section~\ref{sec:planar}. We conclude the paper with open problems for future work.



\section{$D = \{1, 2\}$ and either  $S = \varnothing$ or $S=\{1\}$}\label{sec:Dsmall}



\begin{proposition}\label{prop:Dsmall}
 The games \gDistance{\{1,2\}}{S} are \cclass{PSPACE}-hard for $S = \varnothing$ and $S=\{1\}$.
\end{proposition}

\begin{proof}
Since \ruleset{BiGraph-Node-Kayles} is \cclass{PSPACE}-hard, we will construct a  reduction from a bipartite graph $G = (V_B\cup V_R, E)$ on which  \ruleset{BiGraph-Node-Kayles} is played to a  graph  $G'$ on  which  \gDistance{\{1,2\}}{S} is played, where $S = \varnothing$ or $S = \{1\}$.
 
We first look at the reduction from
$G$ to $G'$ when $S=\{1\}$, which is illustrated in \cref{fig:d12s1}.  Note that we do not show the edges connecting the sets $V_B$ and $V_R$ to better focus on the reduction, which preserves $G$ as a subgraph. 

\begin{figure}[h!t]
  \begin{center}\includegraphics[scale=.5]{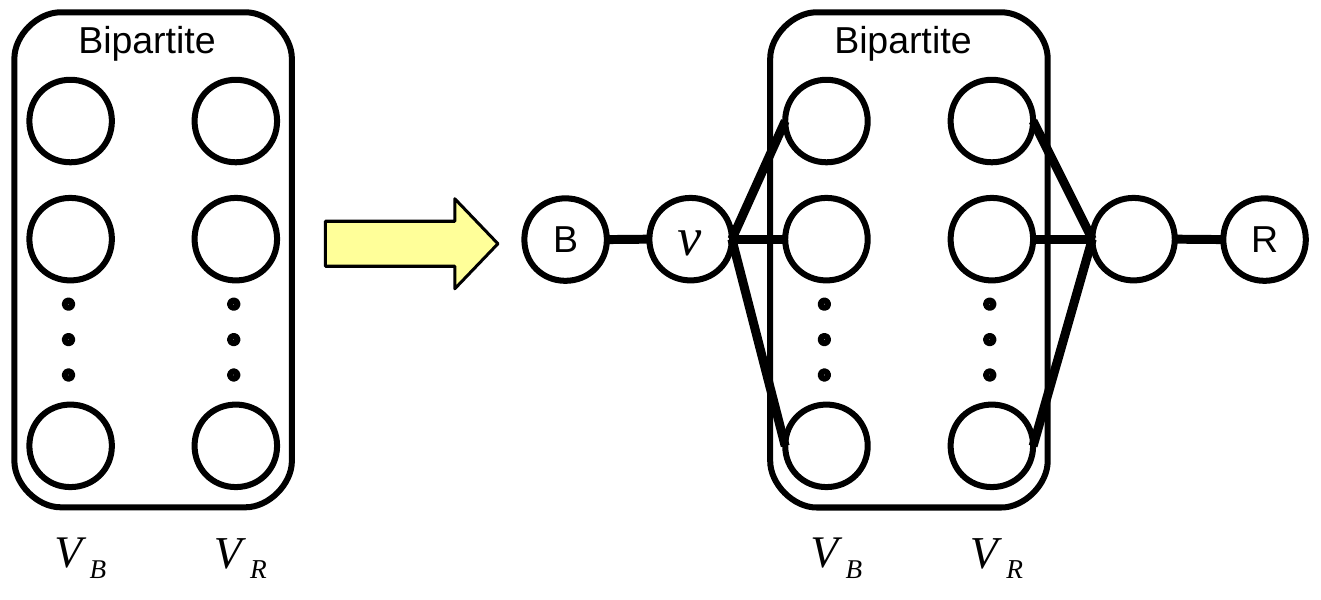}\end{center}
  \caption{The reduction from \ruleset{BiGraph-Node-Kayles} on the left  to \gDistance{\{1, 2\}}{\{1\}} on the right.}
  \label{fig:d12s1}
\end{figure}

In \ruleset{BiGraph-Node-Kayles}, $V_B$ is restricted to only be playable by Blue and $V_R$ to only be playable by Red. Our goal is to create the graph $G'$ via a play-for-play reduction so that playing \gDistance{D}{S} on $G'$ is equivalent to playing \ruleset{BiGraph-Node-Kayles} on $G$.  We describe the reduction as it relates to the vertices in $V_B$.

The first goal is to ensure that $V_B$ cannot be played by Red, hence we connect all vertices from $V_B$ to an uncoloured vertex (labeled $v$) that is connected to an external vertex coloured blue. Since this vertex is distance two from all vertices in $V_B$, no vertex in $V_B$ can be coloured red. Also, the vertex $v$ is at distance one from B, so it can be coloured neither red nor blue, so the vertex condition is satisfied.  
Next we check the edge condition. We need to be careful because vertex $v$ now connects vertices in $V_B$ with a path of length two. However, all vertices in $V_B$ can only be colored in blue, and since $S=\{1\}$, these new paths do not impose restrictions on vertices in the graph $G$.  Replicating the gadget on the right-hand side of the bipartite graph with B replaced by R completes the reduction for $S=\{1\}$. 

We now consider the case $S = \varnothing$, which removes the constraint on labeling the vertex $v$ blue, so we need to put additional vertices into the gadget to create this restriction.  We replace $v$ by a path of length two consisting of two uncoloured vertices $v_1$ and $v_2$ and a terminal vertex coloured red, as shown in \cref{fig:d12}. The red vertex now keeps vertex $v_2$ from being coloured blue.

\begin{figure}[h!]
  \begin{center}\includegraphics[scale=.45]{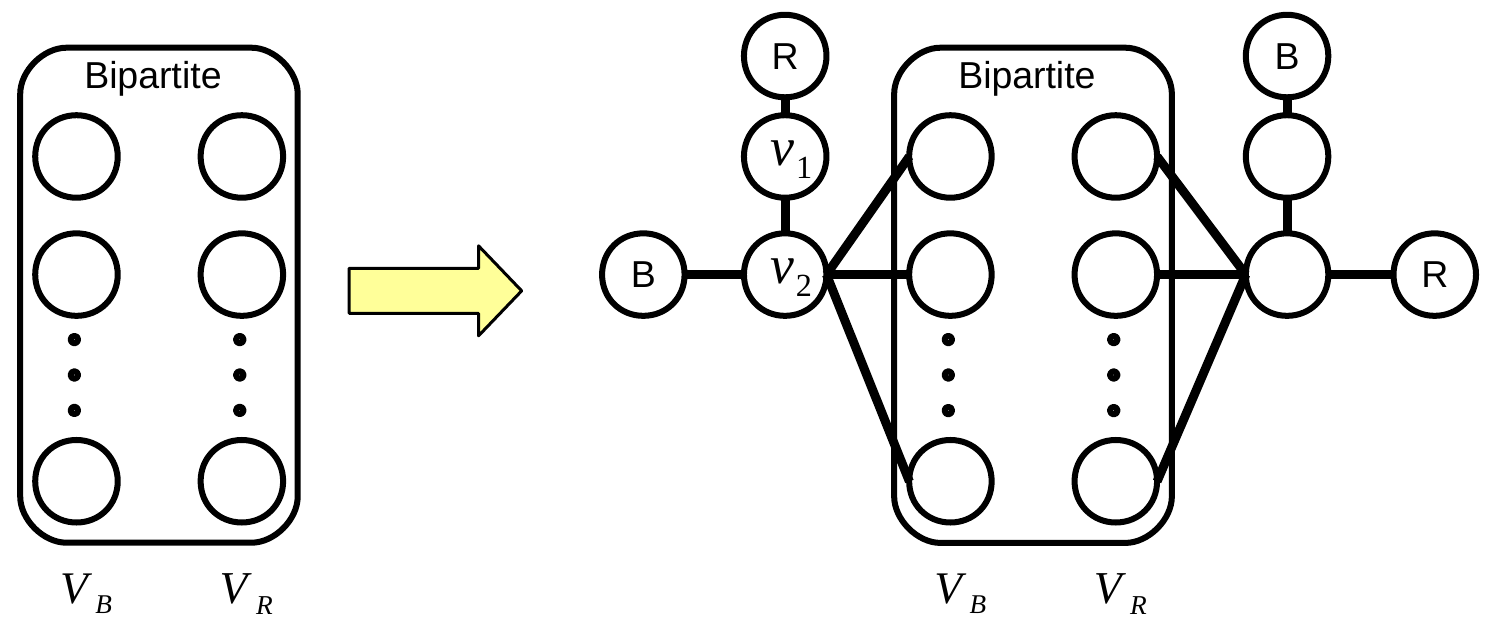}\end{center}
  \caption{The reduction from \ruleset{BiGraph-Node-Kayles} on the left to \gDistance{\{1, 2\}}{\varnothing} on the right.}
  \label{fig:d12}
\end{figure}

In addition, the intermediate vertex $v_1$ is unplayable by both players as it is distance two from B and hence cannot be coloured red, and is adjacent to R and so cannot be coloured blue. Therefore the gadget consisting of these four vertices satisfies the vertex condition. Since the restrictions on colouring vertices in the same colour have been removed completely, none of the paths of length two created by $v_1$ has any impact, so the edge condition is satisfied as well. We replicate this gadget on the right-hand side, switching the roles of R and B for $V_R$, which completes the reduction.  
\end{proof}


Note that in these reductions for $D=\{1,2\}$ we inserted a fixed number of vertices (four and eight, respectively), but that the number of additional edges is a function of the number of vertices of the original graph $G$. Specifically, we added a total of $n+2$ and $n+6$ edges, respectively. The number of edges of $G$ does not play any role in the construction of the gadget. This is quite different from the general case discussed later, where the number of the gadgets (and hence the number of vertices and edges) inserted also depends on the number of edges  of the original graph $G$.

\section{Construction of the forbidden path gadget}\label{sec:gadgets}

For the remaining sets $S$ and $D$ to be considered in this paper, we will utilize a common  construction for the various reductions. So far we have only added vertices and edges, but in the general case, we will also replace edges by subgraphs.  As before, the concern is to make any vertex that  is added into the graph unplayable by each of the two players in such a way that the vertices in the original graph $G$ from which we reduce are not affected. We will achieve this by creating a forbidden vertex gadget and a forbidden path gadget. 

\begin{lemma}\label{lem:gadget}
  If  $D$ or $S$ equals the set $\{1, 2, \ldots, r\}$ for some $r$, and the other is a subset of $\{1, 2, \ldots, r\}$, then we can create a {\em forbidden vertex gadget} $F(r)$ {\em of size $r$} which creates a vertex $v$ such that $v$ is uncoloured, but neither player may choose to play at $v$, and the playability of any vertex connected to $v$ is not affected by the vertices in the gadget. Furthermore, all vertices in the gadget are either coloured, or uncoloured and unplayable.
\end{lemma}

\begin{figure}[h!]
\begin{center}
  \includegraphics[scale=.5]{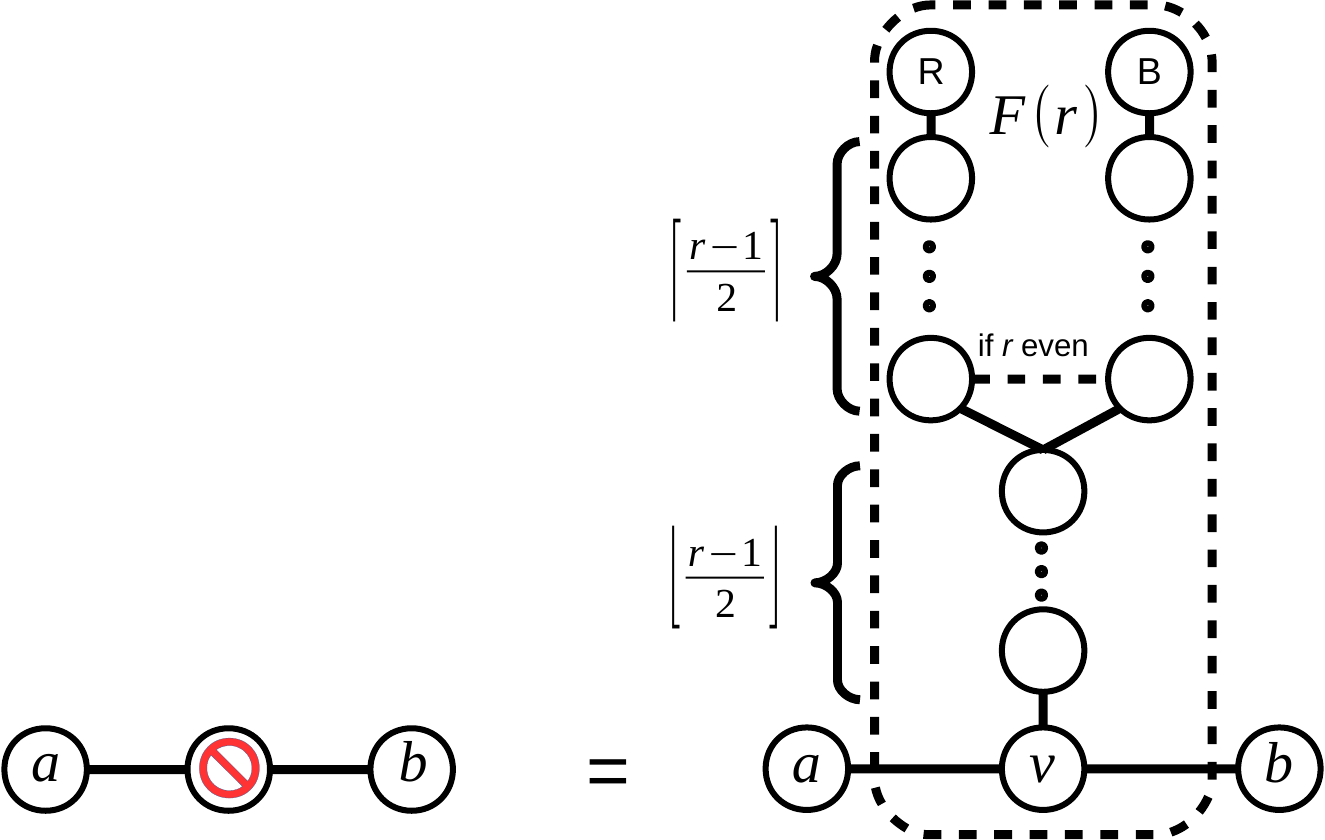}
  \end{center}
  \caption{Given $D = \{1, 2, \ldots, r\}$, $F(r)$ creates a forbidden vertex (between $a$ and $b$) without affecting the playability of $a$ or $b$.}
  \label{fig:unplayableVertex}
\end{figure}

\begin{proof}

  Consider the gadget $F(r)$ shown in  \cref{fig:unplayableVertex}, which is connected to vertices $a$ and $b$. We now prove that any uncoloured vertex in the gadget is unplayable by either player.  Since the gadget is symmetric, we assume without loss of generality that $D= \{1, 2, \ldots, r\}$. For any $r$, the paths from the vertices labeled R and B, respectively, to vertex $v$ are of length $r$, so each of the vertices on these paths cannot be coloured blue and red, respectively. This means the vertices common to both paths cannot be coloured with either red or blue. In addition, any vertex connected to $v$ is not affected by the  vertices labeled R and B as their distance from such a vertex is at least $r+1$.  For the upper portions of the two paths we now need to ensure that these vertices also cannot be coloured with either colour. When $r$ is even, the shortest path from R to B using the dashed edge has length $2\lceil \frac{r-1}{2}\rceil+1=r+1$, so each of those vertices is within distance $r$ of the R and B vertices and cannot be coloured in either colour. When $r$ is odd, then the shortest path from R to B has length $2\lceil \frac{r-1}{2}\rceil+2=r+1$. Overall, all the unlabeled vertices in $F(r)$ cannot be played by either player, as stated.
\end{proof}

We will use a path of an appropriate length made up of forbidden vertex gadgets in the various reductions, where the size of the forbidden vertex gadget is equal to the maximal element in the distance set that consists of consecutive integers.  We refer to a path consisting of $t$ forbidden vertices $F(r)$ as $F\!P(t,r)$, as shown in \cref{fig:forbiddenPath}. 

\begin{figure}[h!]
  \begin{center}\includegraphics[scale=.5]{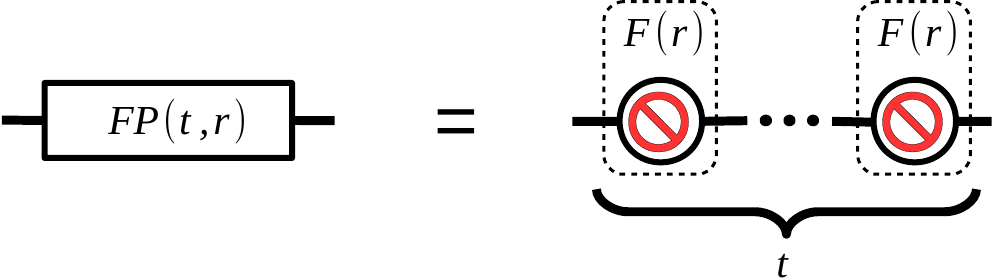}\end{center}
  \caption{$F\!P(t,r)$ represents a path of $t$ forbidden vertices of size $r$ with edges leaving either end.}
  \label{fig:forbiddenPath}
\end{figure}

Such a path will be used to replace an edge between two vertices of the original graph, as shown in \cref{fig:d1234nBipartiteEdge}. We will refer to this operation as {\em edge replacement}. Note that inserting either one of these gadgets into the graph $G$ automatically satisfies the vertex condition of the play-for-play reduction by \cref{lem:gadget}.

  \begin{figure}[h!]
  \includegraphics[scale=.6]{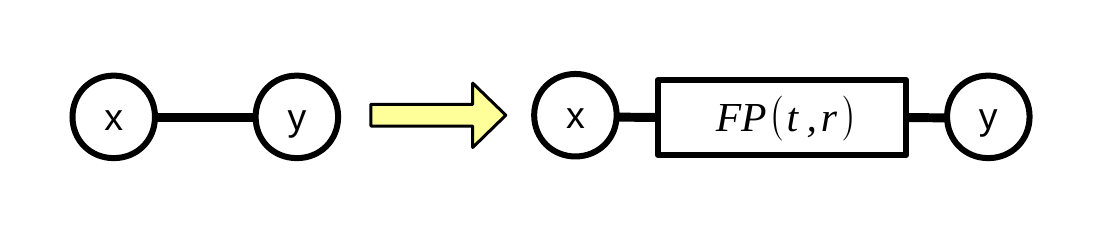}
  \caption{The edge replacement operation for edges $(x,y)$ in the original graph $G$ for the reduction to graph $G'$ for \gDistance{D}{S}, where  $t$ and $r$ depend on the particular reduction.}
  \label{fig:d1234nBipartiteEdge}
\end{figure}


Note that when we replace an edge with a forbidden path $F\!P(t,r)$, we add a total of $t(1+\lfloor \frac{r-1}{2} \rfloor +2\lceil\frac{r-1 } {2} \rceil +2)\approx t\frac{3}{2}r$ vertices. The number of edges added is of the same order (there is a difference of 1 edge depending on whether $r$ is odd or even). Overall, replacement of an edge by a forbidden path $F\!P(t,r)$ or addition of $F\!P(t,r)$ to any vertex adds O$( t r)$ edges and vertices to the graph. We will use this fact when looking at specific reductions in the next few sections.


We are now ready to prove that \gDistance{D}{S} is \cclass{PSPACE}-hard for more general sets $D$ and $S$.

\section{$D$ or $S$ equals $\{1, 2, \ldots, r\}$,\,   $\max(D) \ne \max(S)$}\label{sec:max not same}


Let us first assume that $D=\{1, 2, \ldots, r\}$ with $r \ge 2$ and consider the case  $S = \varnothing$. These distance games are generalizations of \ruleset{Snort}. We define \ensnort{r} to be the game \gDistance{\{1, 2, \ldots, r\}}{\varnothing} for $r \ge 2$.

\begin{proposition}\label{th:ensnorthard}
 \ensnort{r} is  \cclass{PSPACE}-hard.
\end{proposition}

\begin{proof}
\ensnort{r} is a generalization of \ruleset{Snort}, which is \cclass{PSPACE}-hard \cite{Schaefer1978}, so we will use a reduction from \ruleset{Snort} to prove the result. Let $G = (V, E)$ be any graph. 
Since colouring a vertex in \ensnort{r}  affects vertices up to distance $r$ from a coloured vertex, we need to create a reduction that allows us to increase the distance between the vertices in $G$ in such a way that any vertex that is inserted is not playable by either player. This can be achieved by 
replacing each edge in $G$ with a forbidden path $F\!P(r-1,r)$. 
Now playing \ensnort{r} on $G'$ is  exactly the same as playing \ruleset{Snort} on $G$.  
\end{proof}


For non-empty sets $S$ with $\max(S) < r=\max(D)$, the same reduction as in the case $S=\varnothing$ works because we only used properties of $D$, specifically the maximal reach, in the construction of the forbidden vertices $F(r)$ and paths $F\!P(r-1,r)$. As long as $\max(S) < r$, colouring restrictions from the set $S$  do not impact any of the  uncoloured vertices from $V$ in $G'$, as all vertices from $V$ are now at distance $r$ from any other vertex in $V$.

\begin{corollary}\label{thm:SsmallD}
 For $r \ge 2$, \gDistance{ \{1, 2, 3, \ldots, r\}}{S} is \cclass{PSPACE}-hard when $\max(S)<r$.
\end{corollary}

The case  $S=\{1,2,\ldots,r\}$ is very similar to the one treated above, with the roles of $S$ and $D$ interchanged, with reduction from \ruleset{Col}, which is \cclass{PSPACE}-hard as well \cite{Fenner2015}.

\begin{proposition}
\label{encolhard}
 For $r \ge 2$, \gDistance{D}{\{1, 2, 3 \ldots, r\} } is \cclass{PSPACE}-hard when $\max(D) < r$.
\end{proposition}

To measure the impact of the gadget insertion on the graph in both cases, we let $m = |E|$ denote the number of edges in the original graph $G$. Since we have replaced each edge by a forbidden path $F\!P(r-1,r)$,  we have added O$(m r^2)$ edges and vertices. 
Note that the maximal value of $m$ is $n^2$ for a complete graph, and that $r \le n$. So in the worst case scenario, this reduction adds $O(n^4)$ edges and vertices.

We now turn to the question of why we have to exclude the case $\max(S)=\max(D)$. 
If $\max(S)=\max(D)=r$, then for  a vertex $x$  coloured in one colour, a  vertex $y$ such that $(x,y)\in E$ would now be uncolourable in either colour, not just the other colour. This is why we need a 
reduction from a game that has the feature that a vertex adjacent to a coloured vertex in $V$ cannot be coloured in either colour. This suggests a reduction from \ruleset{Node-Kayles}.

\section{$D$ or $S$ equals $\{1, 2, \ldots, r\}$ and $r=\max(D)=\max(S)$ }\label{sec:same max}

In this section, we consider distance games in which the maximum distances not playable by the same and different colours are identical. 

\begin{proposition}
    \gDistance{D}{S} is \cclass{PSPACE}-hard when either $D$ or $S$ equals $\{1, 2, \ldots, r\}$ and the other set is a subset of $\{1, 2, \ldots, r\}$ with $2 \le r =\max(D)=\max(S)$.
\end{proposition}


\begin{proof}
    Let $G = (V, E)$ be any graph. We reduce from \ruleset{Node-Kayles}, which is \cclass{PSPACE}-hard \cite{Schaefer1978},  
    and start with the extreme case where both $D$ and $S$ consist of the full set $\{1,2\ldots, r\}$. Since all vertices at distances less than or equal to $r$ are unplayable by either player,  we replace each edge $(x,y) \in E$ by the path gadget $F\!P(r-1,r)$. 
Then playing \gDistance{D}{S} on $G'$ is exactly the same as playing \ruleset{Node-Kayles} on $G$. 
    
    For the more general case where $S \ne D$, the same construction works as any vertex inserted through the path gadget is unplayable as long as one of the two sets equals $\{1,2,\dots,r\}$ by \cref{lem:gadget}. The conclusion follows as in the case $S=D$ since the only relevant distances for play are the maximal distances. 
\end{proof}

Since we used the same edge replacement as in Proposition~\ref{th:ensnorthard}, we add at most O$(n^4)$ edges and vertices in this case as well.

\section{Distance Games on Planar Graphs}\label{sec:planar}

So far we have considered the computational complexity of distance games on any graph. A game on a more specialized (potentially simpler) graph may be easier to solve, and therefore might not be \cclass{PSPACE}-hard even though the game played on a general graph is \cclass{PSPACE}-hard.

In \cite{Burke2018} the authors show that \ruleset{Snort} and \ruleset{Col} are \cclass{PSPACE}-hard on planar graphs. The edge replacement operation we have used in our various reductions results in a planar graph $G'$ when starting from a planar graph $G$. Thus 
our constructions show that the corresponding planar \gDistance{D}{S} games are also \cclass{PSPACE}-hard. The stronger results are listed below:


\begin{proposition}
  Both planar \gDistance{ \{1, 2, 3, \ldots, r\}}{S} and planar \gDistance{D}{\{1, 2,  \ldots, r\} } are \cclass{PSPACE}-hard when $\max(S) < r$ and   $\max(D) < r$, respectively.

\end{proposition}

\section{Conclusion and Future Work}

To summarize, we used various play-for-play reductions from known \cclass{PSPACE}-hard games to show that \gDistance{D}{S} is \cclass{PSPACE}-hard when $D$ or $S$ is $\{1,2,\ldots, r\}$ and the other one is a subset. The games  from which we reduced were in most cases the natural choices based on the properties of the distance sets $D$ and $S$. 
At the heart of the reductions was the forbidden vertex gadget. To obtain a play-for-play reduction required that the larger of the distance sets consists of consecutive integers. This leads to the following question:

\begin{openProblem}
 Is  \gDistance{D}{S}  \cclass{PSPACE}-hard for cases not covered by our results?
\end{openProblem}



As illustrated in the previous section, planar \ruleset{Snort} and planar \ruleset{Col} being \cclass{PSPACE}-hard implies that many of the distance games we considered are also \cclass{PSPACE}-hard on planar graphs because our reduction preserves planarity. If this is also the case for \ruleset{Node-Kayles}, then these results can be further extended. Thus we are interested in:


\begin{openProblem}
  Is planar \gDistance{D}{S}  \cclass{PSPACE}-hard for other cases?  
\end{openProblem}

The case where max$(S) = $ max$(D)$ and at least one of $S$ or $D$ is equal to $\{1, \ldots, r\} $ would be covered using our reduction if \ruleset{Node-Kayles} were \cclass{PSPACE}-hard on planar graphs.




\bibliography{biblio}
\bibliographystyle{abbrv}

\end{document}